%% file: main.tex
\documentclass{sig-alternate}

\usepackage{etoolbox}
\makeatletter
\patchcmd{\maketitle}{\@copyrightspace}{}{}{}
\makeatother

\usepackage{color}
\usepackage{enumitem}
\usepackage{multirow} 
\allowbreak

\graphicspath{{figures/}}
\DeclareGraphicsExtensions{.pdf,.eps}

\usepackage[font=bf]{caption}
\usepackage[labelformat=simple]{subcaption}

\usepackage{booktabs}
\usepackage{mathtools}

\newtheorem{theorem}{Theorem}

\newtheorem{lemma}{Lemma}
\newtheorem{corollary}{Corollary}

\usepackage{prettyref}
\newrefformat{th}{Theorem~\ref{#1}}
\newrefformat{cor}{Corollary~\ref{#1}}
\newrefformat{fig}{Figure~\ref{#1}}
\newrefformat{tab}{Table~\ref{#1}}
\newrefformat{eq}{Eq.~\eqref{#1}}

\renewcommand{\P}{\mathbb{P}}
\newcommand{\E}{\mathbb{E}}

\newcommand{\ind}[1]{\textup{\textbf{1}}\{#1\}}

\newcommand{\CA}{\textup{CA}}
\newcommand{\RW}{\textup{RW}}
\newcommand{\figwidth}{0.7\columnwidth}

\begin{document}

\title{Duration and Intensity of Cumulative\\
Advantage Competitions}

\numberofauthors{1} 

\author{
\alignauthor{Bo Jiang\textsuperscript{\textasteriskcentered}, Liyuan Sun\textsuperscript{\textdagger}, Daniel R. Figueiredo$^\ddagger$, Bruno Ribeiro\textsuperscript{\textsection}, Don Towsley\textsuperscript{\textasteriskcentered}}\\
\vspace{0.5em}
\affaddr{\textsuperscript{\textasteriskcentered}University of Massachusetts, Amherst, MA, USA}\\
\vspace{0.5em}
\affaddr{\textsuperscript{\textdagger}Tsinghua University, Beijing, China}\\
\vspace{0.5em}
\affaddr{$^\ddagger$Federal University of Rio de Janeiro, Brazil}\\
\vspace{0.5em}
\affaddr{\textsuperscript{\textsection}Carnegie Mellon University, Pittsburgh, PA, USA}
}

\maketitle
\input{TEX/abstract}




\input{TEX/intro}

\input{TEX/related}

\input{TEX/model}

\input{TEX/results}

\input{TEX/proofs}

\input{TEX/conclusion}


%
\bibliographystyle{abbrv}
\bibliography{refs}
%



\balancecolumns

\end{document}

%% file: TEX/abstract.tex
\begin{abstract}
The role of skill (fitness) and luck (randomness) as driving forces on the dynamics of resource accumulation 
in a myriad of systems have long puzzled scientists. Fueled by undisputed inequalities that emerge 
from actual competitions, there is a pressing need for better understanding the effects of skill and luck in 
resource accumulation. When such competitions are driven by externalities such as cumulative advantage (CA), the rich-get-richer effect, 
little is known with respect to fundamental properties such as their duration and intensity. In this work we 
provide a mathematical understanding of how CA exacerbates the role of luck in detriment of skill in simple 
and well-studied competition models. 
We show, for instance, that if two agents are competing for resources that arrive sequentially at each time unit, 
an early stroke of luck can place the less skilled in the lead for an extremely long period of time, a phenomenon we call ``struggle of the fittest''. In the absence of CA, the more skilled quickly prevails despite any early stroke of luck that the less skilled may have. 
We prove that duration of a simple skill and luck competition model exhibit power 
law tails when CA is present, regardless of skill difference, which is in sharp contrast to exponential tails when 
CA is absent. Our findings have important implications to competitions not only in complex social systems but 
also in contexts that leverage such models. 
\end{abstract}

%% file: TEX/intro.tex
\section{Introduction}
\label{sec:intro}


Resource (wealth) accumulation in a variety of complex
systems lead to remarkable inequalities in resource distribution.
The connectivity skewness of autonomous system~\cite{Mahanti13}, webpage
and hashtag popularity~\cite{Barab99,Broder2000,Kwak2010}, and the number
of friends and followers in online social
networks~\cite{Kumar2010,Kwak2010,Ugander2011} have profound implications on the performance of information systems, such as caching~\cite{barford1999} and searching~\cite{adamic2001}.
In these complex systems, skill (fitness) and luck (randomness) are believed to be two fundamental ingredients that drive resource  
accumulation dynamics in a variety of social and complex systems. Apart from them, externalities are also believed 
to play a fundamental role in how resources accumulate.
In this regard, cumulative advantage (CA), where accumulated resources promote gathering even 
more resources,\footnote{This phenomenon appears in the literature under many variants such as Price's
cumulative advantage model~\cite{SollaPrice76}, preferential attachment \cite{Barab12,Barab99,Bianc01}, ``the rich get richer'', Matthew effect \cite{diprete2006cumulative,Merton68,Perc2014}, and path-dependent increasing returns \cite{Arthur}.}
is a simple and widely-applied framework that captures the essence of network externalities, 
while random walk (RW) serves as the hallmark framework without externalities. 

These two parsimonious models have a long tradition in providing key 
insights into the process of resource accumulation, with CA, for instance, presented 
as a mechanism to help explain the presence of power-laws observed in empirical data~\cite{Barab99,SollaPrice76,Perc2014,Rosen1981}.
It is also well-known that in both frameworks the most skilled is bound to accumulate more resources. However, the assurance that the most skilled eventually prevails provides little practical insight. The question is, how long will it take for the most skilled to prevail?



We approach this problem by considering classical, simple and well-studied theoretical models for competitions 
based on skill and luck that are either coupled with or free of cumulative advantage. We focus on competitions between two agents and study two fundamental 
aspects of competitions: {\em duration} -- the time required for the most skilled to overtake its competitor and 
forever enjoy undisputed leadership; {\em intensity} -- the number of times competitors tie for the leadership. In 
this direction, we make the following main contributions.

\begin{itemize}
\item In the case where the two competitors have equal fitness, we obtain the asymptotic tail distributions for both duration and intensity of CA competitions. We demonstrate that they are power laws with respective tail exponents $-1/2$ and $-1$, which are independent of the initial wealth of the competitors.

\item  In the case where the two competitors have unequal fitness, we derive asymptotic lower and upper bounds for the tail distribution of duration of CA competitions, and an upper bound for the tail distribution of their intensity. These bounds show that duration is heavy tailed while intensity is exponential tailed in the presence of CA. In particular, duration is heavier tailed while intensity is lighter tailed than corresponding RW competitions.

\item We observe that a slight difference in fitness results in a extremely heavy tail for duration of CA competitions. Thus, a slightly more skilled individual might have to hang on to the competition for an extremely long period of time before taking the ultimate lead, a phenomenon we call the ``struggle of the fittest''. 
\end{itemize}

Despite 90 years since the basic CA model was first proposed~\cite{Polya23}, 
known as P\'olya's urn model, our work is, to the best of our knowledge, the first to characterize 
the duration and intensity distributions of CA competitions with skill. We believe our results have 
profound implications to our understanding of competitions, beyond its importance to the performance of 
systems that leverage resource distribution.

The rest of the paper is organized as follows. Section \ref{sec:related} briefly discusses the related work. Section \ref{sec:model} introduces the CA competition model. Section \ref{sec:results} presents the theoretical results, illustrated and supplemented by simulations. Section \ref{sec:proofs} provides the proofs for the theoretical results in Section \ref{sec:results}. Section \ref{sec:conclusion} concludes the paper.

%% file: TEX/related.tex

\section{Related work}
\label{sec:related}

Resource accumulation is an ubiquitous phenomenon that naturally arises in a variety of social and complex
systems. The problem is usually framed as a competition among agents for resources that are abundant, and
has been studied in different contexts across various disciplines ranging from proteins binding within
a cell~\cite{PRL_Eisenberg03,Kampen:1981vs} to views of online social media~\cite{borghol2011,figueiredo2011} and citations among scholarly papers~\cite{SollaPrice76,Wang:2013to}.

Models proposed for resource accumulation competitions are generally driven by skill (fitness)
and luck (randomness) as well as externalities, such as network effects. The P{\'o}lya's urn model~\cite{Polya23,Mahm08} is widely used to capture these effects.
Most previous work on P{\'o}lya's urn model and its generalizations focuses on the share of resources
gathered by each agent, also known as the agent's {\em market share}, proving convergence
and limiting results of the market share distribution \cite{Mahm08,PT_Janson05,Peman07,JCPC_Oliveira08}.

However, two fundamental metrics associated with competitions, {\em duration} - how long it takes for
the undisputed winner to emerge, and {\em intensity} - how many times the competitors tie for the
leadership, have largely been neglected in the literature. Previous results establish that the most skilled agent eventually wins~\cite{Mahm08},
and that average intensity up to time $t$ is approximately $(\log t)^\alpha$, where $\alpha$ depends
on the relative skill of the competitors~\cite{JSM_GL10,JSM_GGL10}.
To the best of our knowledge, no previous work has provided rigorous characterizations for the distributions of duration
and intensity of competitions in P{\'o}lya's urn models. Our work partially fills this gap for the two competitor case and sheds light on some recent approximate results~\cite{JSM_GL10,JSM_GGL10}.

%% file: TEX/model.tex

\section{Models}
\label{sec:model}

In this section, we formally introduce competition models for two competing agents and give precise definitions for two fundamental metrics of a competition, i.e.\ its duration and intensity.

\subsection{General Setup and Metrics}
\label{subsec:metrics}

Let $X$ and $Y$ denote the two agents that engage in the competition. Each agent is associated with a positive \emph{fitness} value that reflects its intrinsic competitiveness or skill level. Let $f_X$ and $f_Y$ denote the fitness of $X$ and $Y$, respectively, and $r = f_X/f_Y$ the fitness ratio. Without loss of generality, we assume that $f_X \geq f_Y$ and hence $r \geq 1$. 

The resource that the agents compete for will be generically referred to as wealth, which is measured in discrete units. The competition starts at time $t=0$ with agents $X$ and $Y$ having $x_0$ and $y_0$ units of initial wealth, respectively. We consider a discrete-time process.  At each time step, one unit of wealth is added to the
system and given to either $X$ or $Y$. Denote by $X_t$ and $Y_t$ the respective cumulative wealth of $X$ and $Y$ at time $t$.
The complete history of the competition $\{(X_t,Y_t)\}_{t=0}^\infty$ then forms a discrete-time discrete-space stochastic process. The state space $S$ is the first quadrant of the integral lattice (see \prettyref{fig:sample-path}),  
\[
S = \{(x,y)\in \mathbb{Z}^2:x\geq 1,  y\geq 1\}.
\]
The initial condition is $(X_0,Y_0) = (x_0,y_0)$. How the process evolves over time is defined by specific competition models, of which the CA competition model to be introduced in Section \ref{subsec:CA-model} is an example.

\begin{figure}
\centering
\includegraphics[width=0.5\columnwidth]{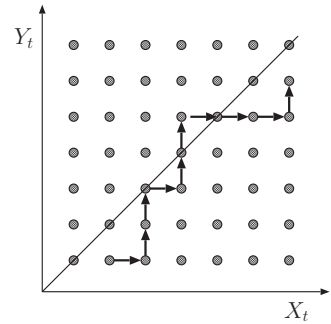}
\caption{State space of competition processes, with an illustration of a sample path with $x_0=2$, $y_0=1$, and` three ties at time $t=3,5,7$.}
\label{fig:sample-path}
\end{figure}

We now make the notions of duration and intensity of competitions more precise by defining them through events of wealth ties. Given a competition process $\{(X_t, Y_t)\}_{t=0}^\infty$, we say that a \emph{tie} occurs at time $t$ if $X_t = Y_t$. \prettyref{fig:sample-path} shows three ties at time $t=3,5,7$. 

The {\em duration} $T$ of a competition is defined
to be the time of the last tie, i.e.,
\[
T = \sup \{t\geq 0: X_t = Y_t\}.
\]
When there is no tie, we follow the standard convention that $T = \sup\emptyset = -\infty$. The competition ends at time $T$ in the sense that one of the agents takes the lead and never lose it again after $T$.

The {\em intensity} $N_t$ of a competition until time $t$ is the number of ties that occur by time $t$, i.e.,
\[
N_t = \sum_{i=0}^t \ind{X_i = Y_i}, 
\]
where $\ind{A}$ is the indicator of event $A$. The \emph{intensity} $N$ of a competition is the total number of ties throughout the competition, 
i.e., $N = \lim_{t\to \infty} N_t$. This measures the intensity of the competition in the sense that it counts the number of potential changes in leadership.
Note that $T<+\infty$ if and only if $N<+\infty$.

\subsection{CA Competition Model}
\label{subsec:CA-model}

In the CA competition model,  the unit of wealth introduced at time $t+1$ is given to $X$ with probability
\[
p_{X,t} = \frac{f_X X_t}{f_X X_t + f_Y Y_t} = \frac{r X_t}{r X_t + Y_t};
\]
otherwise it is given to $Y$.  Note that the transition probability $p_{X,t}$ embodies both fitness and CA effects (externalities). 

More formally, in the CA competition model, the complete history $\{(X_t,Y_t)\}_{t=0}^\infty$ forms a discrete-time Markov chain with stationary transition probabilities. The transition probability $\P[(X_{t+1},Y_{t+1})=(x',y')\mid (X_t,Y_t)=(x,y)]$
is given by
\begin{equation}\label{eq:CA-transition-prob}
Q_{\CA,r}(x,y;x',y') = \begin{dcases*}
\frac{rx}{rx+y}, & if $(x',y') = (x+1,y)$,\\
\frac{y}{rx+y}, & if $(x',y') = (x,y+1)$,\\
0,& otherwise.
\end{dcases*}
\end{equation}
Note that the transition probabilities are spatially inhomogeneous, i.e., they depend on the current state $(x,y)$, which makes the analysis difficult, especially when $r > 1$.


For the purpose of comparison, a RW competition model  
incorporates skill and luck but not the CA effect (no externalities), where the transition probabilities are determined entirely by the fitness ratio $r$.
In particular, the probability that agent $X$ receives the unit of wealth introduced at any time is always given by
\[
p_{X} = \frac{f_X}{f_X + f_Y} = \frac{r }{r + 1}.
\]
Thus the RW competition model is a discrete-time Markov chain with the same state space $S$ as the CA competition model, but with the following spatially homogeneous transition probabilities,
\[
Q_{\RW,r}(x,y;x',y') = \begin{dcases*}
\frac{r}{r+1}, & if $(x',y') = (x+1,y)$,\\
\frac{1}{r+1}, & if $(x',y') = (x,y+1)$,\\
0,& otherwise.
\end{dcases*}
\]
The spatial homogeneity of the transition probabilities leads to a more tractable analysis. In fact, the difference process $\{X_t-Y_t\}$ is a standard biased RW with parameter $r/(r+1)$.  Thus the abundance of known results for 
RW \cite{RW} can be directly translated into results for RW competitions, including duration and intensity as we have defined in Section \ref{subsec:metrics}.

Throughout the rest of the paper, we use CA$_{=}$ and RW$_{=}$ to denote CA and RW competitions with identical fitness ($r=1$), respectively.
We use CA$_{\neq}$ and RW$_{\neq}$ to denote CA and RW competitions with distinct fitnesses ($r>1$), respectively.
Before presenting our results, we point out here some connections between the CA and RW models that are useful in our analysis. In particular, in CA$_{=}$, all paths connecting two given states $(x_0,y_0)$ and $(x,y)$ have the same probability. This is a nice property that CA$_{=}$ shares with RW,  which enables us to leverage existing results on RW in our analysis of CA$_{=}$.
Unfortunately, this property is lost in CA$_{\neq}$, where we resort to the Chapman-Kolmogorov equation for upper and lower bounds on the probabilities of interest. In the limiting case where $X_t$ and $Y_t$ are both large but comparable to each other, the connection to RW is again partially retained, a fact we also exploit in the analysis of CA$_{\neq}$.

%% file: TEX/results.tex

%
%

\section{Results} \label{sec:results}

In this section we present our theoretical results for duration and intensity distributions, which 
are also illustrated graphically and supported by extensive numerical simulations. \prettyref{tab:result} provides a summary of our main results along with prior knowledge about RW competitions from the literature. Note that $\P_{\langle\textup{\sc Model}\rangle,r}^{(x_0,y_0)}$ denotes the probability in model $\langle\textup{\sc Model}\rangle\in \{\CA, \RW\}$ with fitness ratio $r$ and initial state $(x_0,y_0)$. 
The following notations have been used in \prettyref{tab:result} and will be used throughout the rest of the paper.
\begin{itemize}
\item $f(x)\sim g(x)$ if and only if $\lim_{x\to\infty} f(x) / g(x) = 1$.
\item $f(x)\lesssim g(x)$ if and only if $\limsup_{x\to\infty} f(x) / g(x) \leq 1$.
\item $f(x)\gtrsim g(x)$ if and only if $\liminf_{x\to\infty} f(x) / g(x) \geq 1$.
\end{itemize}
All proofs are relegated to Section \ref{sec:proofs}.

\begin{table}[t]
\renewcommand{\arraystretch}{1.3} 
\centering
\begin{tabular}{@{}lccl@{}}
\toprule
metric   & $\langle$\textsc{Model}$\rangle$ & $r=1$ & \quad $r > 1$  \\
\midrule
\multirow{3}{*}{\begin{tabular}{@{}l@{}}duration $T:$\\$\P_{\langle\textup{\sc Model}\rangle,r}^{(x_0,y_0)}[T \ge t]$\end{tabular}} & \multirow{2}{*}{CA} & \multirow{2}{*}{$\sim t^{-1/2}$} & $\lesssim t^{-(r-1)x_0}$ \\
&  &  & $\gtrsim t^{-(r-1)(x_0-\frac{1}{r})}$ \\
\cmidrule{2-4}
&  RW & $1$  & $\leq \left[ \frac{4r}{(r+1)^2} \right]^t$ \\
\midrule
\multirow{2}{*}{\begin{tabular}{@{}l@{}}intensity $N:$\\$\P_{\langle\textup{\sc Model}\rangle,r}^{(x_0,y_0)}[N \ge n]$ \end{tabular} }& CA & $\sim n^{-1}$      & $\leq \left(\frac{2}{r+1}\right)^{n-1}$ \vspace{0.02in} \\
\cmidrule{2-4}
 & RW & $1$ & $ \left(\frac{2}{r+1}\right)^{n-1}$ \\
\bottomrule
\end{tabular}
\caption{
Tail distributions for duration and intensity of competitions in both RW and CA models. Multiplicative constants are omitted in all expressions involving $t$ and $n$. 
The RW statistics can be found in most textbooks on the topic, e.g.~\cite[pp.\ 113,116]{RW}.}
\label{tab:result}
\end{table}


\subsection{Competition Duration}
\label{subsec:duration}

As shown in \prettyref{tab:result}, a RW$_{=}$ competition never ends, i.e., $\P_{\RW,1}^{(x_0,y_0)}[T = \infty] = 1$, while the duration of a RW$_{\neq}$ competition exhibits an 
exponential tail with a base that is inversely proportional to the fitness ratio $r$, which means that RW$_{\neq}$ competitions are 
generally very short. The story for CA competitions is drastically different. The introduction of CA 
guarantees that a competition always ends, i.e. $\P_{\CA,r}^{(x_0,y_0)}[T<\infty] = 1$, even when the two agents are equally fit, which is in sharp contrast with endless
RW$_{=}$ competitions. On the other hand, CA fundamentally {\em increases} duration of competitions between unequally fit agents, which always has a power-law distribution,
in contrast with a sub-exponential distribution for RW$_{\neq}$. Thus, cumulative advantage does \emph{not} always make competitions shorter as one might expect.

\subsubsection{Equal Fitness Case: CA$_{=}$}
The following theorem shows that the duration $T$ for CA$_{=}$ is heavy-tailed with an asymptotic
power-law distribution.
\begin{theorem}\label{thm:duration-equal}
The duration of a CA$_{=}$ competition has the following asymptotic tail distribution, 
\begin{equation}\label{eq:duration-equal}
\P_{\CA,1}^{(x_0,y_0)}[T\geq t] \sim \frac{1}{2^{x_0 + y_0 - 5/2} \sqrt{\pi}B(x_0,y_0)} t^{-1/2},
\end{equation}
where $B(x,y) = \int_0^1 s^{x-1} (1-s)^{y-1}ds$ is the beta function.
\end{theorem}
It follows from \eqref{eq:duration-equal} that
\[
\P_{\CA,1}^{(x_0,y_0)}[T<\infty] = 1-\lim_{t\to\infty} \P_{\CA,1}^{(x_0,y_0)}[T\geq t] = 1,
\] 
i.e. the duration of CA$_{=}$ is almost surely finite. 

Note, however, that the power-law exponent is always $-1/2$, independent of the initial wealth $x_0$ and $y_0$.
Consequently, although the duration of CA$_{=}$ is finite rather than infinite as in RW$_{=}$, the 
expected duration is still infinite, even if $x_0$ is significantly larger than $y_0$ or vice versa.

On the other hand, the initial wealth $(x_0,y_0)$ does affect the location of the distribution. \prettyref{fig:duration-equal} shows the duration distributions from simulations for various values of initial wealth, with the asymptotes in \prettyref{eq:duration-equal} superimposed. Each simulation curve is the average of $10^5$ independent runs for $10^7$ time steps each. Note the good agreement between theory and simulation in the tails. When both $x_0$ and $y_0$ increase but are kept equal, the distribution curve shifts upwards, which means the competition lasts longer. When the initial wealth of only one agent ($y_0$ here) increases, the distribution curve shifts downwards, which means the competition is shorter.

\begin{figure}[!htb]
\centering
\includegraphics[width=\figwidth]{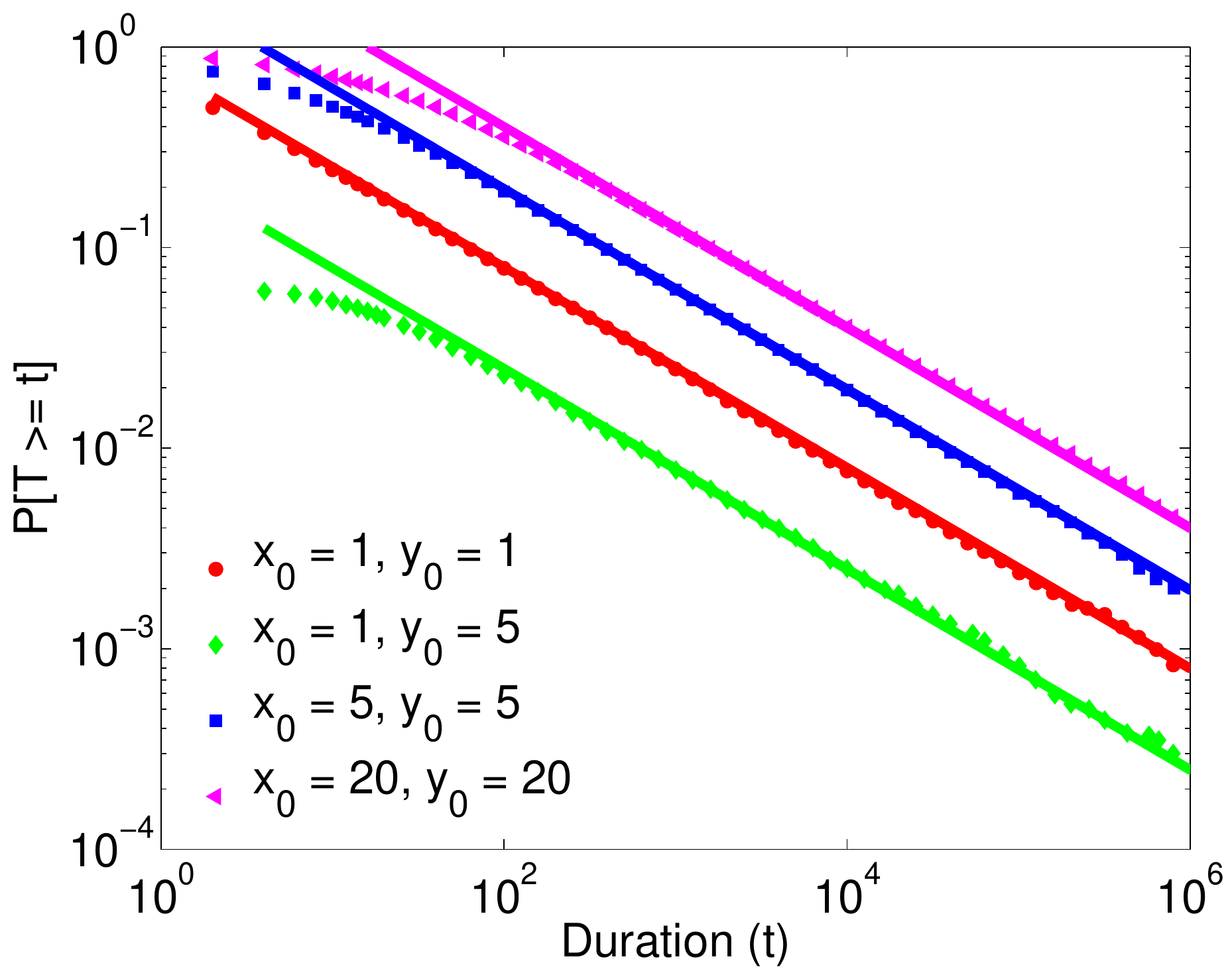}
\caption{{\label{fig:identCA}} Tail distribution for duration of CA$_{=}$  with various $(x_0,y_0)$. The dots are simulation results. 
The solid lines are the asymptotes in Eq.~\eqref{eq:duration-equal}.}
\label{fig:duration-equal}
\end{figure}

\subsubsection{Different Fitness Case: CA$_{\neq}$}

The next theorem shows that the tail distribution of the duration $T$ for CA$_{\neq}$ is asymptotically bounded by power laws from both below and above.

\begin{theorem}\label{thm:duration-unequal}
The tail distribution of the duration of a CA$_{\neq}$ competition has the following asymptotic bounds, 
\begin{equation}\label{eq:duration-unequal}
 \varphi_1 \, t^{-(r-1)x_0} \,  \lesssim \: \P_{\CA,r}^{(x_0,y_0)}[T \ge t] \: \lesssim  \varphi_2 \, t^{-(r-1) (x_0 - 1/r)},
\end{equation}
where 
\begin{equation}\label{eq:phi1}
\varphi_1 = \frac{\Gamma(rx_0+y_0)}{(r+1)x_0 2^{x_0+y_0-1} \Gamma(x_0)\Gamma(y_0)}, 
\end{equation}
and
\begin{equation}\label{eq:phi2}
\varphi_2 = \frac{2^{(r-1)(x_0-r^{-1})}\Gamma(r^{-1})\Gamma(rx_0+y_0)}{(r+1)(x_0-r^{-1})\Gamma(x_0)\Gamma(y_0)},
\end{equation}
where $\Gamma(x) = \int_0^\infty s^{x-1} e^{-s} ds$ is the gamma function.
\end{theorem}

It follows from the lower bound that $\P_{\CA,r}^{(x_0,y_0)}[T<\infty] = 1$ for $r>1$, i.e. the duration for CA$_{\neq}$ is almost surely finite as is for CA$_{=}$. The constants $\varphi_1$ and $\varphi_2$ turn out to be pretty loose, so the bounds are best interpreted as bounds on the tail exponent.

Note that the power-law exponents in the upper and lower bounds depend on $x_0$ but not on $y_0$, and they differ
only by $1-1/r < 1$. In this sense, the shape of the distribution at large $t$ is largely determined by the fitness ratio
and the initial wealth of the fitter agent, while the initial wealth of the less fit plays a much weaker role. This is illustrated in \prettyref{fig:duration-unequal}, which shows the duration distributions from simulations for $r=1.2$ and various values of $(x_0,y_0)$, alongside the lower bounds from \prettyref{eq:duration-unequal} that are shifted closer to the simulation results for easier comparison of the slopes. Each simulation curve is the average of $10^5$ independent runs for $10^9$ time steps each.

\begin{figure}[ht]
\centering
\begin{subfigure}{1\columnwidth}
\centering
\includegraphics[width=\figwidth]{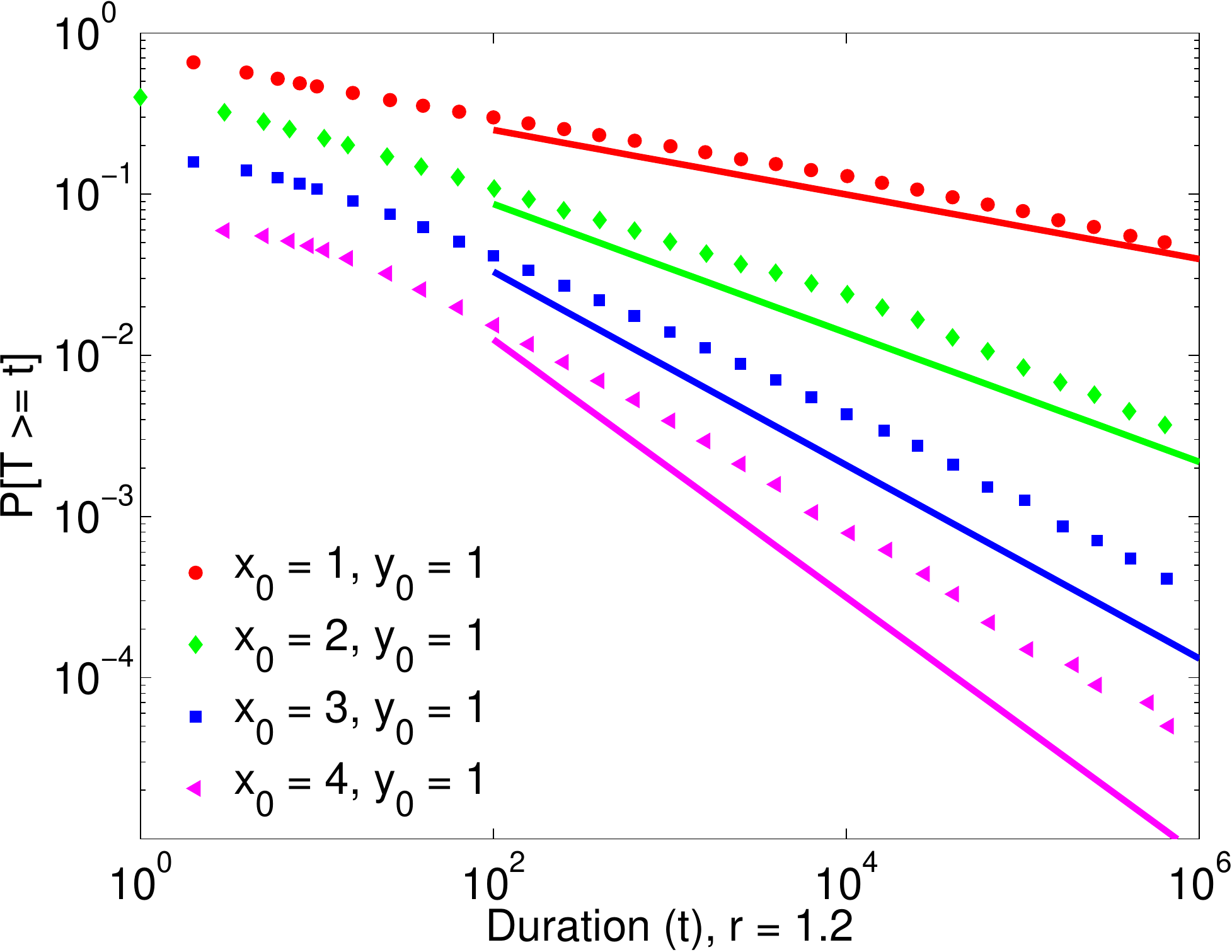}
\caption{}
\label{fig:duration-unequal-x0}
\end{subfigure}
\begin{subfigure}{1\columnwidth}
\centering
\includegraphics[width= \figwidth]{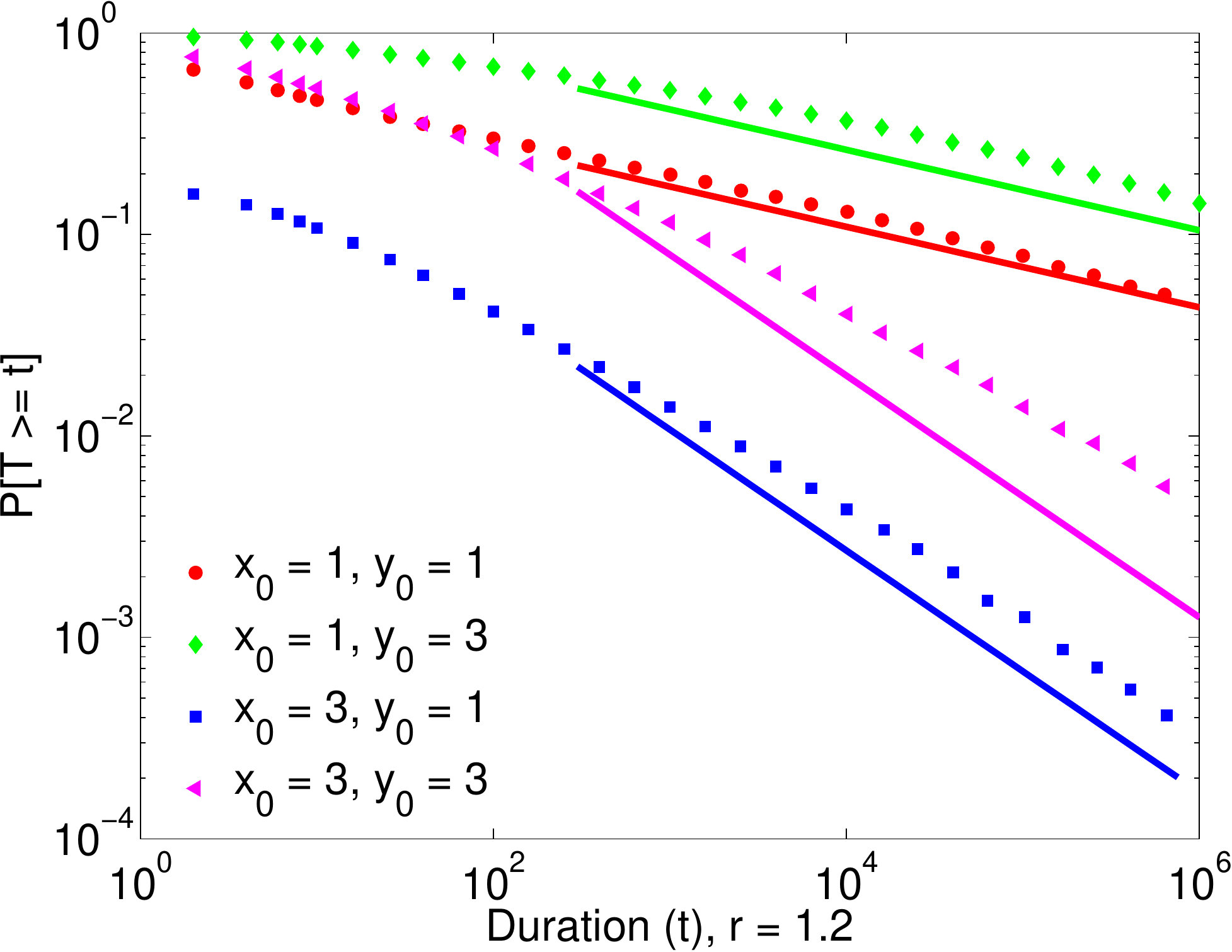}
\caption{}
\label{fig:duration-unequal-y0} 
\end{subfigure}
\caption{
Tail distribution for duration of CA$_{\neq}$ with $r=1.2$ and various $(x_0,y_0)$. The dots are simulation results. 
The solid lines are the asymptotic lower bound in Eq.~\eqref{eq:duration-unequal} but shifted closer to the simulation results for easier visual comparison of the slopes.}
\label{fig:duration-unequal}
\end{figure}

 \prettyref{fig:duration-unequal-x0}
shows how the slopes of the distribution curves, which correspond to the power-law exponents, depend critically on $x_0$. The impact of $x_0$ is two-fold. As $x_0$ increases, the distribution curve becomes more tilted as predicted by the bounds. At the same time, it also shifts downwards. Both changes mean that the competition tends to be shorter.

\prettyref{fig:duration-unequal-y0} shows the impact of changing both $x_0$ and $y_0$. When $x_0$ is fixed, increasing $y_0$ only results in a slight decrease in the absolute value of the slope, in agreement with \prettyref{eq:duration-unequal}. The distribution curve shifts upwards, which means the competition tends to last longer. When both $x_0$ and $y_0$ increase, the situation becomes more intricate. The curve may shift upwards while bending down faster in the tail, which could possibly lead to a crossover in the old and new curves, as is the case of going from $(x_0,y_0)=(1,1)$ to $(x_0,y_0)=(3,3)$. In this case, the new competition is more likely to have a medium long duration.

\subsubsection{Struggle-of-the-Fittest Phenomenon}

Now we look at the impact of fitness ratio $r$ on duration.
Contrasting Eqs.~\eqref{eq:duration-equal} and~\eqref{eq:duration-unequal} leads to an interesting observation.
Departing from CA$_{=}$ by slightly increasing the fitness ratio $r$ from 1 to $1+\varepsilon$, where $\varepsilon$ 
is close to 0, precipitates a {\em significant increase} in the probability of long-lasting competitions, as manifested 
in the discontinuous jump in the power-law exponents from $-1/2$ in Eq.~\eqref{eq:duration-equal} to $-\varepsilon x_0\approx 0$ in Eq.~\eqref{eq:duration-unequal}.
This is opposite to what happens in RW competitions, where a slight increase in fitness departing from RW$_{=}$ to RW$_{\neq}$ 
transforms the competition from one that never ends to one with a geometrically distributed duration.
The lower bound in Eq.~\eqref{eq:duration-unequal} shows that CA$_{\neq}$ with $r < 1+(2x_0)^{-1}$ is more likely to have
long-lasting competitions than CA$_{=}$, despite the fact that the fitter agent is bound to become
the ultimate winner. We call this phenomenon ``struggle of the fittest''.

\prettyref{fig:duration-r}  shows the duration of simulated CA competitions for various fitness ratios $r$. Each simulation curve is the average of $10^5$ independent runs for $10^9$ time steps each. Note how the distribution of duration jumps upward from the curve for CA$_{=}$ to the curve for CA$_{\neq}$ with $r=1.1$. It also shows how the curves for CA$_{\neq}$ become more and more tilted as $r$  increases, being roughly parallel to the CA$_{=}$ curve at $r= 1+(2x_0)^{-1} = 1.5$. 

\begin{figure}[ht]
\centering
\includegraphics[width=\figwidth]{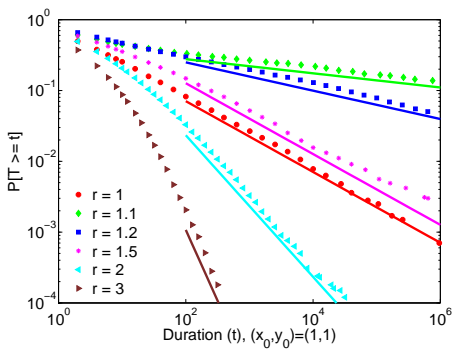}
\caption{
Tail distribution for duration of CA  with various $r$. The dots are simulation results. 
For $r=1$, the solid line is the asymptote in Eq.~\eqref{eq:duration-equal}.
For $r>1$, the solid lines are the lower bound in Eq.~\eqref{eq:duration-unequal} but shifted as in \prettyref{fig:duration-unequal}.}
\label{fig:duration-r}
\end{figure}

\subsection{Competition Intensity}
\label{subsec:intensity}

Given that CA competitions are long-lasting, one might expect them also to be intense, i.e., exhibit many ties ($X_t = Y_t$). As we will see in this section, this intuition is appropriate only for CA$_=$ but not for CA$_{\neq}$.

\subsubsection{Equal Fitness Case: CA$_{=}$}

The following theorem shows that the intensity $N$ of CA$_{=}$ is heavy-tailed with an asymptotic power-law distribution.

\begin{theorem}\label{thm:intensity-equal}
The intensity of a CA$_{=}$ competition has the following asymptotic tail distribution,
\begin{equation} \label{eq:intensity-equal}
\P_{\CA,1}^{(x_0,y_0)}[N\geq n] \sim \frac{1}{2^{x_0+y_0-2} B(x_0,y_0)} n^{-1},
\end{equation}
where $B(x_0,y_0)$ is the beta function as in \prettyref{eq:duration-equal}.
\end{theorem}

In this case, the intensity has infinite expectation, as does the duration. \prettyref{fig:intensity-equal} shows the duration distributions from simulations for various values of initial wealth, with the asymptotes in \prettyref{eq:intensity-equal} superimposed. Each simulation curve is the average of $10^5$ independent runs for $10^7$ time steps each. We observe the same behavior as in \prettyref{fig:duration-equal}.  When both $x_0$ and $y_0$ increase but are kept equal, the distribution curve shifts upwards, which means the competition is more intense. When the initial wealth of only one agent ($y_0$ here) increases, the distribution curve shifts downwards, which means the competition is less intense.

We mention in passing that 
if we have a finite observation time $t_{f}$, the expected intensity $N_{t_f}$ by time $t_f$ grows
as $\log{t_{f}}$, a phenomenon observed for the related CA model in Godr{\`e}che et al.~\cite{JSM_GGL10}.

\begin{figure}[htb]
\centering
\includegraphics[width=\figwidth]{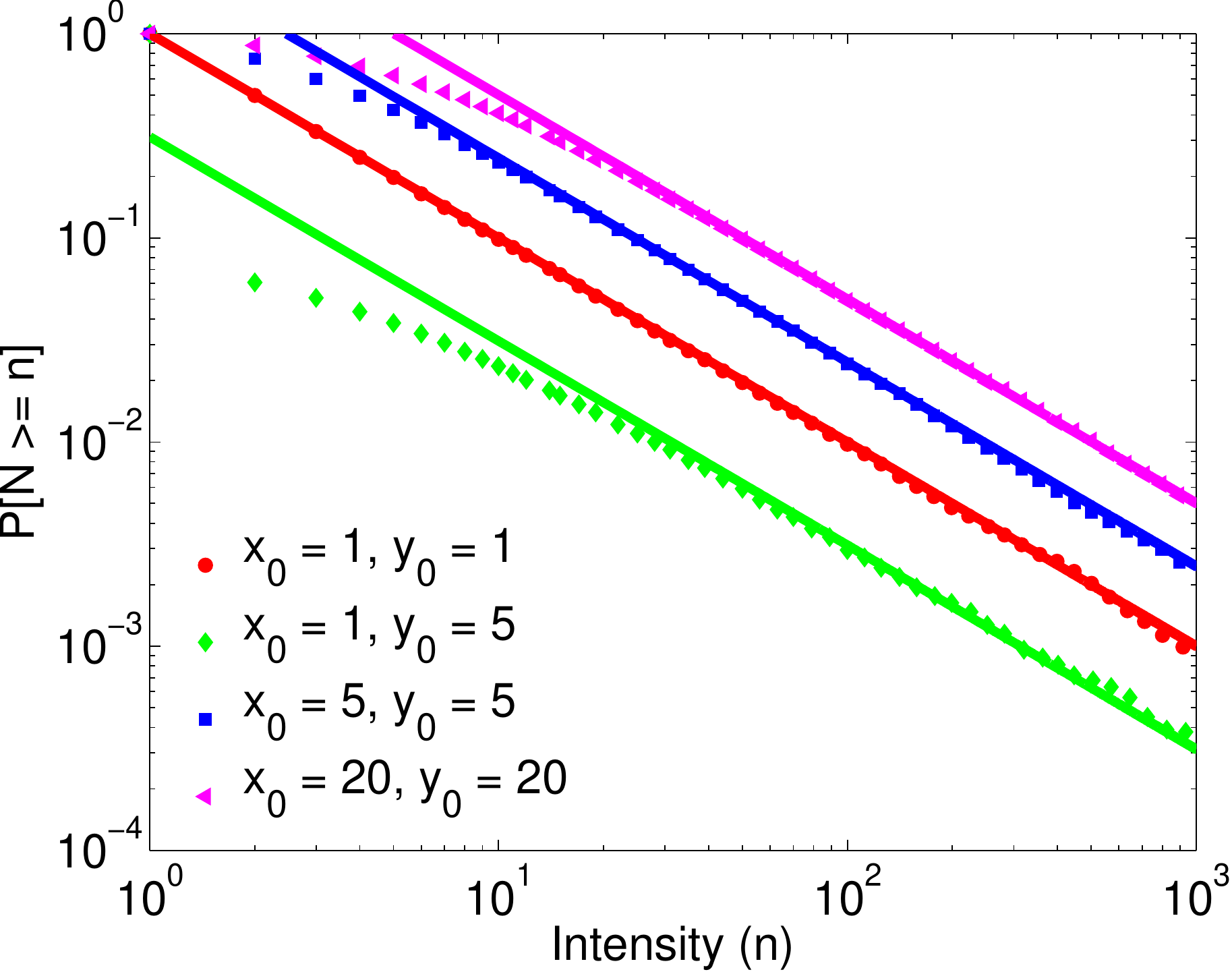}
\caption{Tail distribution for intensity of CA$_{=}$ with various $(x_0,y_0)$. The dots are simulation results. 
The solid lines are the asymptotes from Eq.~\eqref{eq:intensity-equal}. }
\label{fig:intensity-equal}
\end{figure}

\subsubsection{Different Fitness Case: CA$_{\neq}$}

In sharp contrast, CA$_{\neq}$ competitions are not intense despite their long duration. In fact their intensity is
surprisingly mild, bounded above by a geometric distribution, as shown in the next theorem.

\begin{theorem}\label{thm:intensity-unequal}
The tail distribution of the intensity of a CA$_{\neq}$ competition has the following upper bound,
\begin{equation}\label{eq:intensity-unequal}
\P_{\CA,r}^{(x_0,y_0)}[N\geq n] \leq C\left(\frac{2}{1+r}\right)^{n-1},
\end{equation}
with
\[
C =  \begin{cases}
1, & x_0 \leq y_0,\\
\frac{(y_0)_{x_0-y_0}}{(rx_0+y_0)_{x_0-y_0}} \left(1+\frac{1}{r}\right)^{x_0 - y_0}, & x_0 > y_0,
\end{cases}
\]
where $(x)_k = \prod_{i=0}^{k-1} (x+i)$ is the Pochhammer symbol.
\end{theorem}

Note that the expectation and all higher moments of $N$ are finite. 
Therefore, the intensity of CA competitions changes dramatically when
the fitnesses of the two parties become unequal, the distribution shifting from a
power-law tail to an exponential tail. This is illustrated in
Figure \ref{fig:N}, where each simulation curve is the average of $10^5$ independent runs for $10^9$ time steps each. An important observation is that both CA$_{\neq}$\ and RW$_{\neq}$\ competitions 
have intensities that are upper bounded by identical exponential tails (see \prettyref{tab:result}), 
while exhibiting fundamentally different durations.

\begin{figure}[htb]
\centering
\includegraphics[width=\figwidth]{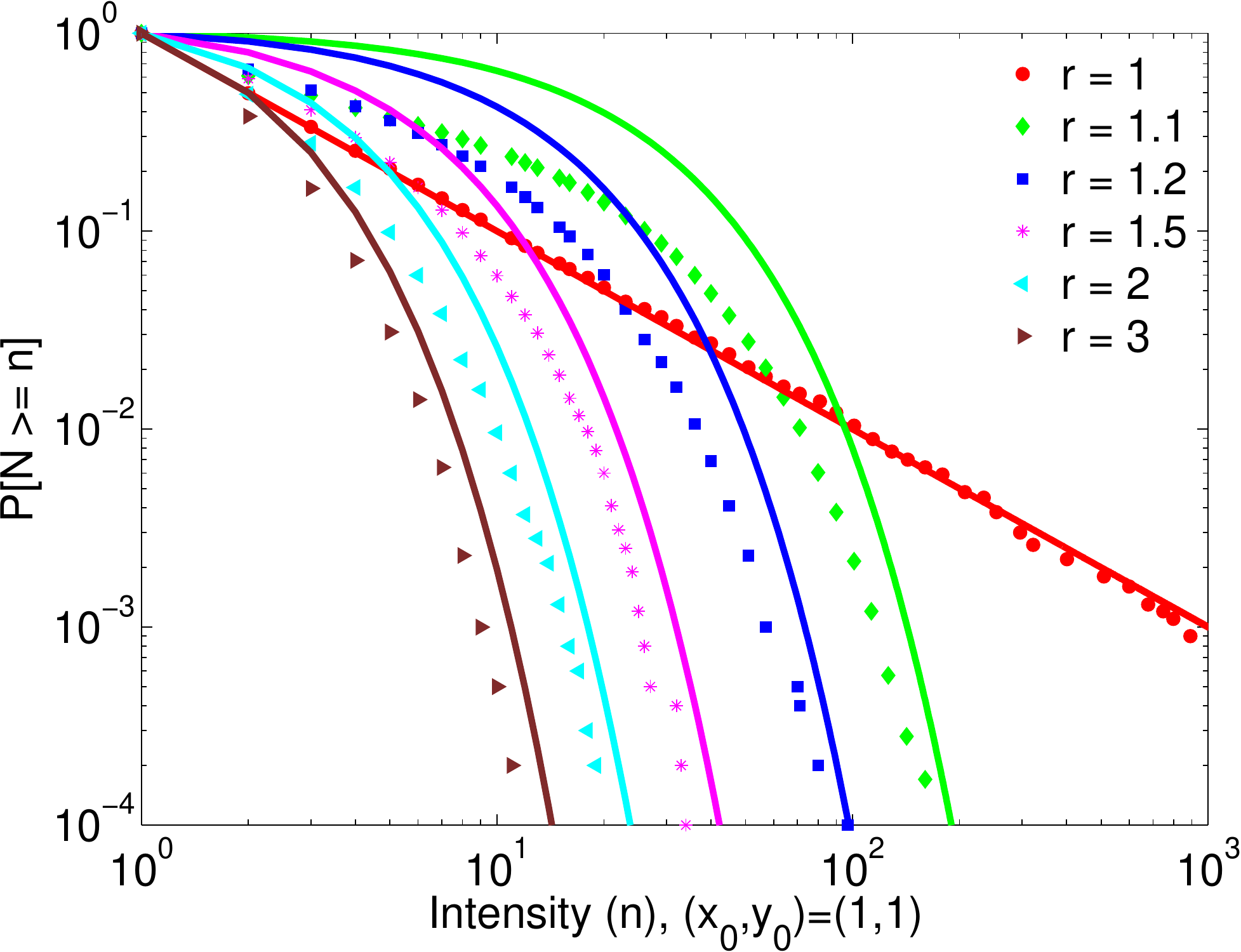}
\caption{
Tail distribution for intensity of CA with various $r$. The dots are simulation results. 
The solid lines are the upper bounds from \prettyref{eq:intensity-unequal} for $r>1$, and the asymptote in \prettyref{eq:intensity-equal} for $r = 1$.}
\label{fig:N}
\end{figure}


Why are CA$_{\neq}$\ competitions simultaneously not intense and long-lasting?
The answer resides in the probability of $Y$ being the eventual winner. In CA$_=$
competitions, $Y$ wins with probability $y_0/(x_0+y_0)$, while in CA$_{\neq}$ competitions $Y$ (the less fit) never wins.
However, for small values of $r$, especially for those very close to one, the dynamics in the initial stages of the competition closely follows that of CA$_{=}$. Thus there is a non-negligible chance that $Y$ takes a significant lead, with the CA effect helping it uphold the lead for a long period of time over which there is no tie.
Eventually, however, fitness effect outweighs the CA effect, and $X$ catches up with $Y$. By then they both have large accumulated wealth, which makes CA$_{\neq}$ behave like RW$_{\neq}$ in the vicinity of $X=Y$, allowing $X$ to quickly establish a lead ahead of $Y$. At this final stage both fitness and CA effects work in favor of $X$, and $Y$ stands little chance in taking the lead again. 
To summarize, the less fit agent has a non-negligible probability of taking an early lead which can last for a very long time 
due to the CA effect, but it will ultimately surrender the lead to the fitter agent and never lead again, a phenomenon 
that we call ``delusion of the weakest'', which is the flip-side of ``struggle of the fittest''.

\prettyref{fig:sample-path-long-duration} illustrates this observation by showing sample paths for different values of $r$, 
all generated using the same sequence of random bits. Note that for $r=1$ (identical fitness), $Y$ wins quickly, whereas 
for $r=1.1$ the fitter agent $X$, having trailed behind for a long time, eventually takes over after 69,426 time steps.
Finally, for both $r = 1.2$ and $r = 1.5$ agent $X$ has no trouble quickly winning the competition.
These sample paths showcase the long struggle of the ``slightly'' fitter agent in competitions with CA effects.

\begin{figure}[htb]
\centering
\includegraphics[width=\figwidth]{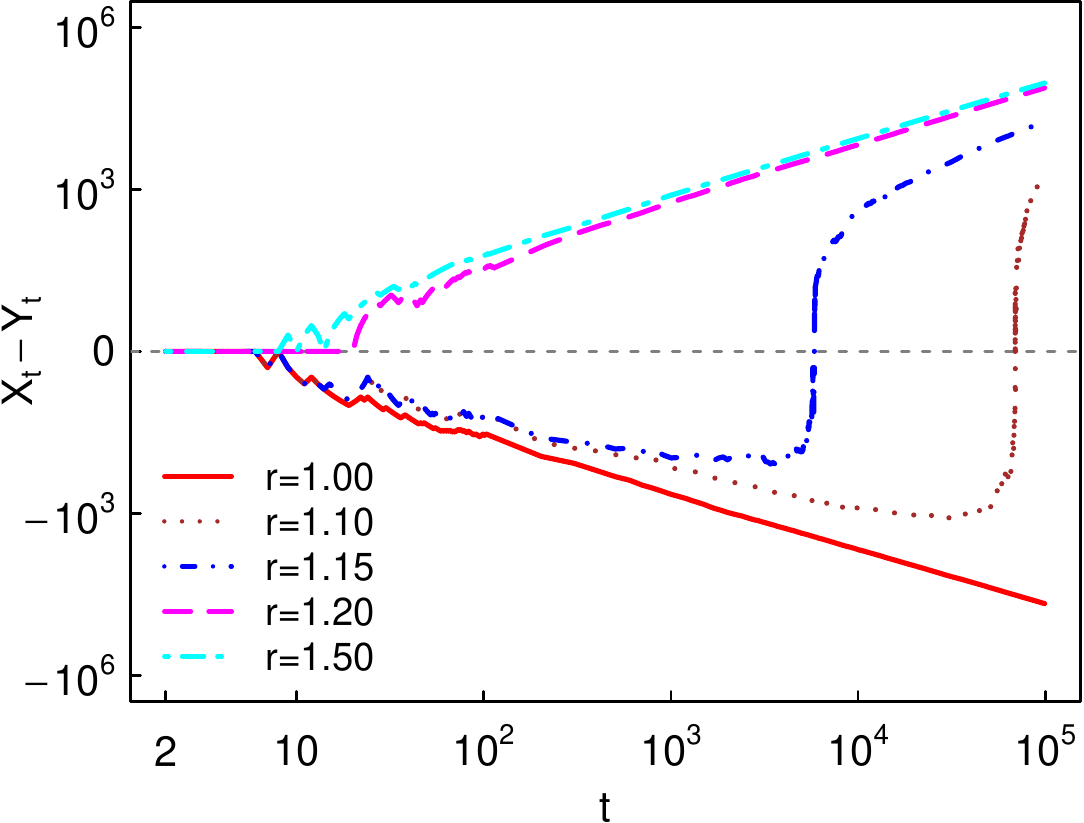}
\caption{``Delusion of the weakest'': sample paths for different values of $r$ ($x_0 =y_0 = 1$), all generated using the same sequence of random bits}
\label{fig:sample-path-long-duration}
\end{figure}

\subsection{Interplay of Duration and Intensity}

In this section, we study the relationship between duration and intensity. 
Note that duration gives a natural upper bound $N \leq T/2$ for intensity, i.e.,
the number of ties is at most half of the duration in any competition.
In CA$_{=}$, duration and intensity are strongly and positively correlated.
In fact, a tie at time $t$ increases the probability of having another tie at a time
later than $t$. More precisely, \cite{Antal10} shows that for CA$_{=}$,
\begin{equation}
\P_{\CA,1}^{(x_0,y_0)}[T > t | X_t = Y_t] \simeq 1-\frac{1}{\sqrt{\pi X_t}}.
\label{eq:P}
\end{equation}
Since $X_t \sim t/2$ at a tie, Eq.~\eqref{eq:P} implies that the later a tie occurs, the more likely another tie will 
occur, intuitively explaining why long-lasting competitions are also intense in this case.

\prettyref{fig:T-N} shows a scatter-plot of duration versus intensity from $10^4$ independent runs of CA$_{=}$ competitions with $x_0 = y_0 = 1$, each simulated for $10^9$ time steps. This unveils a strong positive correlation between the two statistics in log-log 
scale (sample Pearson correlation coefficient of 0.94).

\begin{figure}[htb]
\centering
\includegraphics[width=\figwidth]{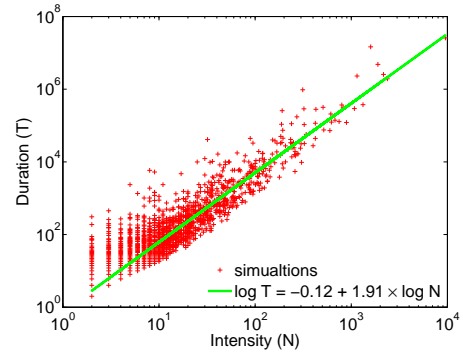}
\caption{Scatter plot of duration vs. intensity for CA$_{=}$. 
}
\label{fig:T-N}
\end{figure}

Interestingly, CA$_{\neq}$ shows a different behavior, since even long-lasting competitions exhibit only a
small number of ties.  \prettyref{fig:EN} shows simulation results for conditional average intensities 
of competitions with $x_0 = y_0 = 1$ and different fitness ratios $r$, conditioned on the
duration being at least $t$. Each simulation curve is obtained from $10^4$ independent runs for $10^9$ time steps each. Note that for $r=1$, the conditional average intensity increases linearly with $t$, but for $r>1$, 
it stabilizes as $t$ increases.
Again, we observe a sharp transition as we move from identical to distinct fitnesses, this time in
the correlation between intensity and duration.

\begin{figure}[htb]
\centering
\includegraphics[width=0.7\columnwidth]{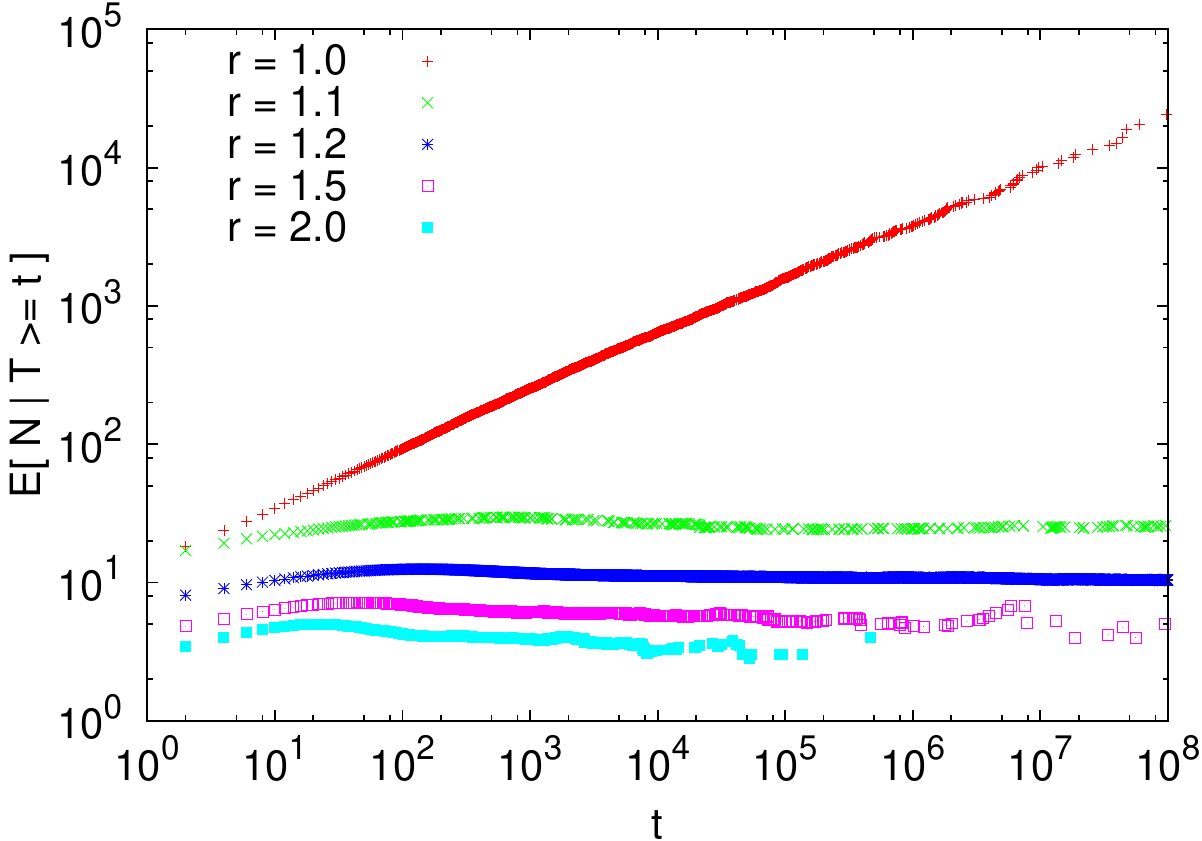}
\caption{Conditional average intensity of competitions conditioned on their duration being at least $t$,
namely $\E_{\CA,r}^{(x_0,y_0)}[N\mid T \geq t]$.
}
\label{fig:EN}
\end{figure}

%% file: TEX/proofs.tex

\section{Proofs} 
\label{sec:proofs}

In this section, we provide proofs for our main results in Section \ref{sec:results}. We will use the following additional notations and definitions.

\begin{itemize}
\item Denote the transition probability from state $(x_0,y_0)$ to state $(x,y)$ in $t=x+y-x_0 - y_0$ steps by
\[
p_r(x_0,y_0;x,y) = \P_{\CA,r}^{(x_0,y_0)}[(X_t,Y_t)=(x,y)].
\]
Note that the transition probability is nonzero only for this specific $t$. Thus we will often omit to mention $t$ explicitly hereafter and assume that the appropriate $t$ has been chosen.

\item
Let $\tau_n$ be the time of the $n$-th tie, which can be defined recursively by $\tau_0 = -\infty$ and
\[
\tau_n  = \inf\{t>\tau_{n-1}:X_t = Y_t\}, \; n \geq 1.
\]
Note that $T = \tau_N$.

\item
Denote by $q_r(x_0,y_0)$ the probability of having no tie after leaving state $(x_0,y_0)$, i.e.,
\[
q_r(x_0,y_0) = \P_{\CA,r}^{(x_0,y_0)}[X_t\neq Y_t, t\geq 1].
\]
Note that $q_r(x_0,x_0) =  \P_{\CA,r}^{(x_0,x_0)}[\tau_2 = \infty]$ and $q_r(x_0,y_0) =  \P_{\CA,r}^{(x_0,y_0)}[\tau_1 = \infty]$ for $x_0\neq  y_0$.

\item
Denote by $A_{n,t}(x,y)$ the set of paths that start from $(x,y)$ at time 0 and end with the $n$-th tie at time $t$, i.e., $\tau_n = t$.
\end{itemize}

\subsection{Proof of \prettyref{thm:duration-equal}}
\label{subsec:proof-duration-equal}

Note that the CA$_{=}$ model is the standard P\'olya urn model. The proof of \prettyref{thm:duration-equal} combines known results for this model. 
Starting from the initial state $(x_0,y_0)$, $X_t$ has a beta-binomial distribution with parameters $x_0$ and $y_0$ \cite{johnson1992univariate}. Note that the event $X_t = Y_t$ occurs only if $t = |x_0 -y_0|+2k$ for some integer $k\geq 0$. For such $t$, $X_t = Y_t$ if and only if $X_t = z_0+k$, where $z_0 = \max\{x_0,y_0\}$. By Eq.~(6.27) of \cite{johnson1992univariate},
\begin{align}\label{eq:tie}
\P_{\CA,1}^{(x_0,y_0)}[X_t=Y_t]&=\P_{\CA,1}^{(x_0,y_0)}[X_t = z_0+k]\nonumber\\
&=  \frac{B(z_0+k,z_0+k)}{B(x_0,y_0)}\binom{t}{k}.
\end{align}
Recall that $q_1(x,y)$ is the probability of having no tie after leaving state $(x,y)$. Thus
\begin{equation}\label{eq:last_tie}
\P_{\CA,1}^{(x_0,y_0)}[T=t] = \P_{\CA,1}^{(x_0,y_0)}[X_t = Y_t] \cdot q_1(z_0+k,z_0+k),
\end{equation}
where the second factor on the right-hand side is the probability of having no tie after $t$.

Recall that the exit probability $E(x,y)$ in \cite{Antal10} is the probability of ever having a tie starting from $(x,y)$, including the initial state $(x,y)$. Thus for $x\neq y$, $q_1(x,y)$ is related to $E(x,y)$ by
\[
q_1(x,y) = 1 - E(x,y).
\]
Using Eq.~(22) of \cite{Antal10} for $E(x,y)$, we obtain
\[
q_1(x+1,x) = q_1(x,x+1) = \frac{\Gamma(x+1/2)}{\Gamma(x+1)\Gamma(1/2)}.
\]
However, $q_1(x,x)\neq E(x,x)=1$. By considering the one-step transition from $(x,x)$ to $(x+1,x)$ or $(x,x+1)$, we obtain
\[
q_1(x,x) = \frac{1}{2} q_1(x+1,x) + \frac{1}{2} q_1(x,x+1)  = \frac{\Gamma(x+1/2)}{\Gamma(x+1)\Gamma(1/2)}. 
\]
Eliminating $\Gamma(x+1/2)$ by the identity 
\[
\Gamma(2x)=\pi^{-1/2}2^{2x-1}\Gamma(x)\Gamma(x+1/2)
\]
in \cite[Eq.~(5.5.5)]{olver2010nist}, and using $\Gamma(x+1)=x\Gamma(x)$ and $\Gamma(1/2)=\sqrt{\pi}$, we obtain
\begin{equation}\label{eq:escape}
q_1(x,x)=\frac{\Gamma(2x)}{x 2^{2x-1}\Gamma(x)\Gamma(x)}=\frac{1}{x2^{2x-1} B(x,x)}.
\end{equation}

Substitution of Eqs.~\eqref{eq:tie} and \eqref{eq:escape} into \prettyref{eq:last_tie} yields
\[
\P_{\CA,1}^{(x_0,y_0)}[T=t] = \frac{1}{B(x_0,y_0)} \cdot \frac{1}{(z_0+k) 2^{2k + 2z_0 - 1}} \binom{t}{k}.
\]
For $t=|x_0-y_0| + 2k$, Stirling's formula yields
\[
\binom{t}{k} \sim \sqrt{\frac{2}{\pi}} t^{-1/2} 2^t,
\]
and hence
\begin{equation}\label{eq:duration-density-equal}
\P_{\CA,1}^{(x_0,y_0)}[T=t] \sim  \frac{1}{2^{x_0 + y_0 - 5/2} \sqrt{\pi}B(x_0,y_0)} t^{-3/2}.
\end{equation}
It is well-known that as $t\to \infty$, $X_t / (X_t + Y_t)$ converges almost surely to a beta random variable $V$. It follows that $|X_t - Y_t| / (X_t + Y_t) \to |2V-1|$. Thus, for $V\neq 1/2$, which holds almost surely, we have $|X_t-Y_t| / (X_t + Y_t) > 0$ for all large enough $t$. Therefore, $\P_{\CA,1}^{(x_0,y_0)}[T=\infty] = 0$.
Summing over $t$ in \prettyref{eq:duration-density-equal}, we obtain as $t\to \infty$,
\begin{align*}
\P_{\CA,1}^{(x_0,y_0)}[T\geq t] & = \sum_{t' = t}^\infty \P_{\CA,1}^{(x_0,y_0)}[T=t'] \\
&\sim \frac{1}{2} \sum_{s=t}^\infty\frac{1}{2^{x_0 + y_0 - 5/2} \sqrt{\pi}B(x_0,y_0)} s^{-3/2}\\
&\sim \frac{1}{2} \int_t^\infty \frac{1}{2^{x_0 + y_0 - 5/2} \sqrt{\pi}B(x_0,y_0)} s^{-3/2} ds \\
&=  \frac{1}{2^{x_0 + y_0 -5/ 2} \sqrt{\pi}B(x_0,y_0)} t^{-1/2},
\end{align*} 
where we have used the fact that half of the terms are zero in the second step, and $\sum_{s=t}^\infty s^{-a} \sim \int_t^\infty s^{-a}ds$ in the third step. This completes the proof of \prettyref{thm:duration-equal}.

\subsection{Proof of \prettyref{thm:duration-unequal}}
\label{subsec:proof-duration-unequal}

Similar to \prettyref{eq:last_tie}, we have
\begin{equation}\label{eq:last_tie_unequal}
\P_{\CA,r}^{(x_0,y_0)}[T=t] = p_r(x_0,y_0;z_0+k,z_0+k) \cdot q_r(z_0+k,z_0+k).
\end{equation}
Thus the proof here amounts to finding expressions for both $p_r(x_0,y_0;z_0+k,z_0+k)$ and $q_r(z_0+k,z_0+k)$ in CA$_{\neq}$. We break the proof into three lemmas.

\begin{lemma}\label{lem:lower_tie}
\begin{equation}\label{eq:lower_tie}
p_r(x_0,y_0;x_0+k,y_0+h) \geq \frac{(x_0)_k (y_0)_h}{(rx_0+y_0)_{k+h}} \binom{k+h}{k},
\end{equation}
for all $k\geq 0, h\geq 0$.
\end{lemma}

\begin{lemma}\label{lem:upper_tie}
\begin{equation}\label{eq:upper_tie}
p_r(x_0,y_0;x_0+k,y_0+h) \leq \frac{(x_0)_k (y_0)_h}{(r^{-1})_k (rx_0+y_0)_h} ,
\end{equation}
for all $k\geq 0, h\geq 0$.
\end{lemma}

\begin{lemma} \label{lem:asymp_escape}
For $r>1$,
\begin{equation}\label{eq:asymp_escape}
q_r(x,x) \to \frac{r-1}{r+1},
\end{equation}
as $x\to\infty$.
\end{lemma}

Before proving these lemmas, we first use them to prove \prettyref{thm:duration-unequal}. 

\begin{proof}[of \prettyref{thm:duration-unequal}]
By \prettyref{lem:lower_tie}, we have
\begin{align*}
& p_r(x_0,y_0;z_0+k,z_0+k) \\
\geq\; & \frac{(x_0)_{k+z_0-x_0} (y_0)_{k+z_0-y_0}}{(rx_0+y_0)_{2k+2z_0-x_0-y_0}}\binom{2k+2z_0-x_0-y_0}{k+z_0-x_0} \\
=\;& \frac{\Gamma(rx_0+y_0)}{\Gamma(x_0)\Gamma(y_0)} \times \frac{\Gamma(k+z_0)\Gamma(k+z_0)}{\Gamma(k+z_0-x_0+1)\Gamma(k+z_0-y_0+1)}\\
&\times \frac{\Gamma(2k+2z_0-x_0-y_0+1)}{\Gamma(2k+2z_0+(r-1)x_0)}
\end{align*}
Using the relation $\Gamma(k+a)/\Gamma(k+b)\sim k^{a-b}$ as $k\to\infty$, we obtain
\begin{align*}
p_r(x_0,y_0;z_0+k,z_0+k) & \gtrsim \frac{\Gamma(rx_0+y_0)}{2^{x_0+y_0-2} \Gamma(x_0)\Gamma(y_0)} (2k)^{-(r-1)x_0-1}\\
&=2(r+1)x_0 \varphi_1 (2k)^{-(r-1)x_0-1},
\end{align*}
where $\varphi_1$ is given by \prettyref{eq:phi1}. Application of this asymptotic bound and  \prettyref{lem:asymp_escape} to \prettyref{eq:last_tie_unequal} yields
\begin{equation}\label{eq:duration-density-unequal}
\P_{\CA,r}^{(x_0,y_0)}[T=t] \gtrsim 2(r-1)x_0 \varphi_1 t^{-(r-1)x_0-1}.
\end{equation}
Note that $\P_{\CA,r}^{(x_0,y_0)}[T=+\infty]=\P_{\CA,r}^{(x_0,y_0)}[N=+\infty]=0$, where the second equality will follow from \prettyref{thm:intensity-unequal}, so we will not provide a separate proof here.
Summing over $t$ in \prettyref{eq:duration-density-unequal} and noting that half of the terms are zero, we obtain as $t\to \infty$,
\begin{align*}
\P_{\CA,r}^{(x_0,y_0)}[T\geq t] & = \sum_{t' = t}^\infty \P_{\CA,r}^{(x_0,y_0)}[T=t'] \\
&\gtrsim \int_t^\infty (r-1)x_0 \varphi_1 s^{-(r-1)x_0-1} ds \\
&=  \varphi_1 t^{-(r-1)x_0},
\end{align*} 
establishing the lower bound.

In a similar way, \prettyref{lem:upper_tie} yields
\begin{align*}
& p_r(x_0,y_0;z_0+k,z_0+k) \\
\lesssim\;  & 2(r+1) (x_0 - r^{-1}) \varphi_2 (2k)^{-(r-1)(x_0-r^{-1})-1},
\end{align*}
where $\varphi_2$ is given by \prettyref{eq:phi2}.
Application of this asymptotic bound and \prettyref{lem:asymp_escape} to \prettyref{eq:last_tie_unequal} yields
\[
\P_{\CA,r}^{(x_0,y_0)}[T=t] \lesssim  2(r-1) (x_0 - r^{-1}) \varphi_2 t^{-(r-1)(x_0-r^{-1})-1},
\]
and
\[
\P_{\CA,r}^{(x_0,y_0)}[T\geq t]  = \sum_{t' = t}^\infty \P_{\CA,r}^{(x_0,y_0)}[T=t'] \lesssim \varphi_2 t^{-(r-1)(x_0-r^{-1})},
\]
establishing the upper bound. 
\end{proof}

Now we prove the lemmas.
Recall that the transition probability $p(x_0,y_0;x,y)$ of going from $(x_0,y_0)$ to $(x,y)$ satisfies the following recursion (Chapman-Kolmogorov equation),
\begin{equation}\label{eq:recursion}
\begin{aligned}
p_r(x_0,y_0;x,y) =\;& \frac{r(x-1)}{r(x-1)+y} p_r(x_0,y_0;x-1,y) \\
&+ \frac{y-1}{rx+y-1} p_r(x_0,y_0;x,y-1),
\end{aligned}
\end{equation}
for $x\geq x_0$, $y\geq y_0$ and $x+y\geq x_0+y_0+1$, with the boundary condition $p_r(x_0,y_0;x,y) = 0$ for $x<x_0$ or $y<y_0$. Note that we have replaced the one-step transition probabilities $Q_{\CA,r}(x-1,y;x,y)$ and $Q_{\CA,r}(x,y-1;x,y)$ by the expressions in \prettyref{eq:CA-transition-prob}.

\begin{proof}[of Lemma \ref{lem:lower_tie}]
We will use the short-hand notation $p(k,h)$ for $p_r(x_0,y_0; x_0+k,y_0+h)$, and $\psi(k,h)$ for the right-hand side of \prettyref{eq:lower_tie}.
We first prove the boundary case for $k = 0$.
By \prettyref{eq:recursion}, for $h\geq 1$,
\[
p(0,h) = \frac{y_0+h-1}{rx_0+y_0+h-1} p(0,h-1) ,
\]
which is a simple recursion in $h$ and can be expanded to yield
\[
p(0,h) = \frac{(y_0)_h}{(rx_0+y_0)_h} p(0,0)= \frac{(y_0)_h}{(rx_0+y_0)_h} = \psi(0,h),
\]
which yields \prettyref{eq:lower_tie} for $k=0$ and $h\geq 1$. Here we have used $p(0,0) = p_r(x_0,y_0;x_0,y_0) = 1$.

Similarly, for the other boundary case $h=0,k\geq 1$, we have
\begin{align*}
p(k,0) &= \frac{(x_0)_k}{(x_0+r^{-1}y_0)_k} p(0,0)= \frac{(x_0)_k}{(x_0+r^{-1}y_0)_k} \\
&\geq \frac{(x_0)_k}{(rx_0+y_0)_k} = \psi(k,0),
\end{align*}
where the last inequality is because $(x)_k$ increases with $x$, and $x_0+r^{-1}y_0\leq rx_0 + y_0$.

For the general case, we use induction on $k+h$. The base case $k+h=1$ is already proven, since either $k=0$ or $h=0$ when $k+h=1$. Assume \prettyref{eq:lower_tie} holds for $k+h=m\geq 1$. Consider $k+h=m+1$. We can also assume $k\geq 1$ and $h\geq 1$, since we have proven the boundary cases for $k=0$ or $h=0$. The recursion in \prettyref{eq:recursion} yields
\begin{align*}
&p(k,h) \\
=\;& \frac{r(x_0+k-1)}{rk+h+c_0-r} p(k-1,h) + \frac{y_0+h-1}{rk+h+c_0-1}p(k,h-1) \\
\geq\;& \frac{r(x_0+k-1)}{rk+h+c_0-1} p(k-1,h) + \frac{y_0+h-1}{rk+h+c_0-1}p(k,h-1),
\end{align*}
where $c_0 = rx_0+y_0$.

Applying the induction hypothesis $p(k-1,h)\geq\psi(k-1,h)$ and $p(k,h-1)\geq\psi(k,h-1)$ to the above inequality  yields
\begin{align*}
& p(k,h)  \\
\geq\;& \frac{r(x_0+k-1)}{rk+h+c_0-1} \psi(k-1,h) + \frac{y_0+h-1}{rk+h+c_0-1} \psi(k,h-1) \\
=\;& \frac{(rk+h)(k+h+c_0-1)}{(k+h)(rk+h+c_0-1)}  \psi(k,h),
\end{align*}
where in the last step we have used
\[
\psi(k-1,h) = \frac{k}{k+h} \cdot \frac{k+h+c_0-1}{x_0+k-1} \psi(k,h),
\]
and
\[
\psi(k,h-1) = \frac{h}{k+h} \cdot \frac{k+h+c_0-1}{y_0+h-1} \psi(k,h).
\]
To complete the proof, it suffices to show that
\[
\frac{(rk+h)(k+h+c_0-1)}{(k+h)(rk+h+c_0-1)} \geq 1,
\]
but this is equivalent to $r\geq 1$, which is true by assumption.
\end{proof}

\begin{proof}[of Lemma \ref{lem:upper_tie}]
The proof of \prettyref{lem:upper_tie} follows the same line of reasoning as that used to prove \prettyref{lem:lower_tie}. The boundary cases can be verified directly. We only outline the induction step here.  Applying \prettyref{eq:upper_tie} to the right-hand side of \prettyref{eq:recursion} yields
\begin{align*}
p(k,h) \leq\;&  \frac{r(x_0+k-1)}{rk+h+c_0-r} \frac{(x_0)_{k-1} (y_0)_h}{(r^{-1})_{k-1} (c_0)_h} \\
&+ \frac{y_0+h-1}{rk+h+c_0-1}\frac{(x_0)_k (y_0)_{h-1}}{(r^{-1})_k (c_0)_{h-1}}\\
=\;&\left[\frac{r(r^{-1}+k-1)}{rk+h+c_0-r}  + \frac{c_0+h-1}{rk+h+c_0-1} \right] \frac{(x_0)_k (y_0)_{h}}{(r^{-1})_k (c_0)_{h}}.
\end{align*}
Note that
\begin{align*}
&\frac{r(r^{-1}+k-1)}{rk+h+c_0-r}  + \frac{c_0+h-1}{rk+h+c_0-1}\\
\leq\; & \frac{r(r^{-1}+k-1)}{rk+h+c_0-r}  + \frac{c_0+h-1}{rk+h+c_0-r} = 1,
\end{align*}
which completes the induction.
\end{proof}

\begin{proof}[of Lemma \ref{lem:asymp_escape}]
Recall that $A_{n,2k}(x,x)$ is the set of paths that start from $(x,x)$ at time 0 and end with the $n$-th tie at time $2k$, i.e., $\tau_n = 2k$. Let $A_{n,2k}=A_{n,2k}(0,0)$. Note that the paths in $A_{n,2k}(x,x)$ are exactly the paths in $A_{n,2k}$ translated by $(x,x)$. Let $\tilde \pi\in A_{n,2k}$ and its state at time $t$ be $\tilde \pi_t = (\tilde x_t,\tilde y_t)$. The translation of $\tilde \pi$ by $(x,x)$, denoted $x+\tilde \pi$, is a path in $A_{n,2k}(x,x)$, whose probability in the CA model is given by
\begin{align*}
\P_{\CA,r}^{(x,x)}[x+\tilde \pi] =\; &\prod_{j=0}^{2k-1} \left(\frac{r(x+\tilde x_j)}{r(x+\tilde x_j) + (x+\tilde y_j)}\right)^{\tilde x_{j+1}-\tilde x_j}\\
&\times \prod_{j=0}^{2k-1} \left(\frac{(x+\tilde y_j)}{r(x+\tilde x_j) + (x+\tilde y_j)}\right)^{\tilde y_{j+1}-\tilde y_j}.
\end{align*}
For fixed $k$ and $\tilde \pi$, as $x\to \infty$, $\P_{\CA,r}^{(x,x)}[x+\tilde \pi]$ converges to
\[
\prod_{j=0}^{2k-1} \left(\frac{r}{r+1}\right)^{\tilde x_{j+1}-\tilde x_j}\left(\frac{1}{r+1}\right)^{\tilde y_{j+1}-\tilde y_j} = \P_{\RW,r}^{(0,0)}[\tilde \pi],
\]
which corresponds to the probability of the path $\tilde \pi$ in a random walk with parameter $r/(r+1)$. Thus, as $x\to\infty$,
\begin{align*}
\P_{\CA,r}^{(x,x)}[\tau_n=2k] &= \sum_{\tilde \pi \in A_{n,2k}} \P_{\CA,r}^{(x,x)}[x+\tilde\pi]\\
&\to \sum_{\tilde \pi \in A_{n,2k}} \P_{\RW,r}^{(0,0)}[\tilde\pi] = \P_{\RW,r}^{(0,0)}[\tau_n=2k].
\end{align*}
After summing over $k$ and using the Dominated Convergence Theorem, we obtain
\begin{align*}
\P_{\CA,r}^{(x,x)}[\tau_n < \infty] &= \sum_{k=1}^\infty  \P_{\CA,r}^{(x,x)}[\tau_n=2k] \\
&\to \sum_{k=1}^\infty  \P_{\RW,r}^{(0,0)}[\tau_n=2k] = \P_{\RW,r}^{(0,0)}[\tau_n < \infty].
\end{align*}
In particular,
\[
q_r(x,x) = 1-\P_{\CA,r}^{(x,x)}[\tau_2 < \infty] \to 1- \P_{\RW,r}^{(0,0)}[\tau_2 < \infty] = \frac{r-1}{r+1},
\]
where we have used Eq.~(3.3) in \cite{RW} for $\P_{\RW,r}^{(0,0)}[\tau_2 < \infty] $ in the last step.
\end{proof}

\subsection{Proof of \prettyref{thm:intensity-equal}}
\label{subsec:proof-intensity-equal}

Recall that $A_{n,t}(x_0,y_0)$ is the set of paths starting from $(x_0,y_0)$ that end with the $n$-th tie at time $t$, i.e., $\tau_n = t$. We will use the short-hand notation $A_{n,t}$ for $A_{n,t}(x_0,y_0)$. As in Section \ref{subsec:proof-duration-equal}, the set $A_{n,t}$ is non-empty only if $t = |x_0-y_0|+2k$ for some integer $k\geq  n -1$, in which case, every path in $A_{n,t}$ ends in state $(z_0+k,z_0+k)$ with $z_0 = \max\{x_0,y_0\}$. Recall from \cite{Antal10} that the probability of any path $\pi$ connecting states $(x_0,y_0)$ and $(x,y)$ is 
\[
\P_{\CA,1}^{(x_0,y_0)}[\pi]=\frac{B(x,y)}{B(x_0,y_0)}=\frac{B(x,y)}{B(x_0,y_0)}2^t \P_{\RW,1}^{(x_0,y_0)}[\pi].
\]
Summing over $\pi \in A_{n,t}$, where $x=y=z_0+k$, we obtain
\begin{equation}\label{eq:nth_tie}
\P_{\CA,1}^{(x_0,y_0)}[\tau_n=t] = \frac{B(z_0+k,z_0+k)}{B(x_0,y_0)} 2^t \P_{\RW,1}^{(x_0,y_0)}[\tau_n=t].
\end{equation}
Thus the probability of having the $n$-th and also the last tie at time $t$ is given by
\begin{align}
&\P_{\CA,1}^{(x_0,y_0)}[T = t, N = n] \nonumber\\
=\;& \P_{\CA,1}^{(x_0,y_0)}[\tau_n=t] \cdot q_1(z_0+k,z_0+k) \nonumber\\
=\;& \frac{1}{2^{x_0+y_0-2}B(x_0,y_0)}\cdot \frac{1}{t+x_0+y_0} \P_{\RW,1}^{(x_0,y_0)}[\tau_n=t]\label{eq:joint}
\end{align}
where we have used Eqs.~\eqref{eq:nth_tie} and \eqref{eq:escape} in the last step.
Note that $\P_{\RW,1}^{(x_0,y_0)}[\tau_n=t]$ is the probability $f_{n,t}(d_0)$ of the $n$-th visit to the origin at time $t$ in a simple symmetric random walk starting from $d_0 = |x_0-y_0|$. Summing over $t$ in \eqref{eq:joint}, we obtain
\begin{align}\label{eq:pmf_N}
&\P_{\CA,1}^{(x_0,y_0)}[N=n] \nonumber\\
=\;& \frac{1}{2^{x_0+y_0-2}B(x_0,y_0)}\sum_{k=n-1}^\infty \frac{1}{2k + d_0+x_0+y_0} f_{n,d_0+2k} (d_0) \nonumber\\
=\;& \frac{1}{2^{x_0+y_0-2}B(x_0,y_0)} G_n(1;d_0),
\end{align}
where
\[
G_n(z;d_0) =  \sum_{k=n-1}^\infty \frac{1}{2k + d_0+x_0+y_0} f_{n,d_0+2k} (d_0) z^{d_0 + 2k}.
\]
To simplify $G_n(z;d_0)$, we have
\begin{align}\label{eq:ode_G}
\frac{d}{dz}[z^{x_0+y_0} G_n(z;d_0)] &= z^{x_0+y_0-1} \sum_{k=n-1}^\infty f_{n,d_0+2k} (d_0) z^{d_0 + 2k} \nonumber\\
&= z^{x_0+y_0-1}  \Phi_n(z;d_0),
\end{align}
where $\Phi_n(z;d_0) = \sum_{k=n-1}^\infty f_{n,d_0+2k} (d_0) z^{d_0 + 2k}$ is the generating function of the probability distribution of the $n$-th visit to the origin in a simple random walk starting from $d_0$. Let $F_1(z)$ be the generating function of the distribution of the time of the first return to the origin in a simple random walk starting from the origin. The standard renewal argument (see e.g. XI.3.d of \cite{fellerOne}) shows that $\Phi_n(z;d_0)$ is given by
\[
\Phi_n(z;d_0) = [\Phi_1(z;1)]^{d_0} [F_1(z)]^{n-1},
\]
where $\Phi_1(z;1)$ and $F_1(z)$ are given by Eqs.~(3.6) and (3.14) of \cite[Chap.~XI]{fellerOne}, respectively. Therefore,
\begin{equation}\label{eq:mgf}
\Phi_n(z;d_0) = z^{-d_0} \left(1-\sqrt{1-z^2}\right)^{n+d_0-1}.
\end{equation}
Substituting \prettyref{eq:mgf} into \prettyref{eq:ode_G} and integrating from 0 to 1 yields
\[
G_n(1;d_0) = \int_0^1 z^{2\min\{x_0,y_0\}-1} \left(1-\sqrt{1-z^2}\right)^{n+d_0-1} dz,
\]
where we have used $x_0+y_0-d_0 = 2\min\{x_0,y_0\}$.
A change of variable $u = \sqrt{1-z^2}$ yields
\[
G_n(1;d_0) = \int_0^1 u (1-u^2)^{\min\{x_0,y_0\}-1} (1-u)^{n+d_0-1} du,
\]
which is upper bounded by
\begin{equation}\label{eq:upper_N}
G_n(1;d_0) \leq \int_0^1 u (1-u)^{n+d_0-1} du = B(2,n+d_0),
\end{equation}
and lower bounded by
\begin{align}\label{eq:lower_N}
G_n(1;d_0) & \geq \int_0^1 u (1-u)^{\min\{x_0,y_0\}-1} (1-u)^{n+d_0-1} du \nonumber\\
& = B\left(2,n+\max\{x_0,y_0\}-1\right),
\end{align}
where we have used $\min\{x_0,y_0\}+d_0 = \max\{x_0,y_0\}$.
Applying Eqs.~\eqref{eq:upper_N} and \eqref{eq:lower_N} to \prettyref{eq:pmf_N} yields
\begin{align*}
\frac{B(2,n+\max\{x_0,y_0\}-1)}{2^{x_0+y_0-2} B(x_0,y_0)} &\leq \P_{\CA,1}^{(x_0,y_0)}[N=n] \\
&\leq \frac{B(2,n+d_0)}{2^{x_0+y_0-2} B(x_0,y_0)}. 
\end{align*}
Note that $\P_{\CA,1}^{(x_0,y_0)}[N=\infty] = \P_{\CA,1}^{(x_0,y_0)}[T=\infty]=0$.
Summing over $n$ and using $\sum_{m=n}^\infty B(2,m) = n^{-1}$, we obtain
\begin{align*}
&\frac{1}{2^{x_0+y_0-2} B(x_0,y_0)}\cdot \frac{1}{n+\max\{x_0,y_0\}-1}\\
\leq\;& \P_{\CA,1}^{(x_0,y_0)}[N\geq n] \leq \frac{1}{2^{x_0+y_0-2} B(x_0,y_0)}\cdot \frac{1}{n+d_0},
\end{align*}
which immediately yields \prettyref{eq:intensity-equal}.

\subsection{Proof of \prettyref{thm:intensity-unequal}}
\label{subsec:proof-intensity-unequal}

We first prove the following lemma.
\begin{lemma}\label{lem:eventual_tie}
The probability $\P_{\CA,r}^{(x_0,y_0)}[\tau_1<\infty]$ of ever having a tie is bounded as follows,
\[
\P_{\CA,r}^{(x_0,y_0)}[\tau_1<\infty] \leq  \begin{cases}
1, & x_0 \leq y_0,\\
\frac{(y_0)_{x_0-y_0}}{(rx_0+y_0)_{x_0-y_0}} \left(1+\frac{1}{r}\right)^{x_0 - y_0}, & x_0 > y_0.
\end{cases}
\]
\end{lemma}
Note that $\P_{\CA,r}^{(x_0,y_0)}[\tau_1<\infty]$ is the exit probability $E(x,y)$ in \cite{Antal10} when $r=1$.

\begin{proof}
The case $x_0 = y_0$ is trivial since $\tau_1 = 0$. When $x_0 <  y_0$, Theorem 3.21 of \cite{janson2004functional} yields $Y_t/X_t \to 0$ almost surely, from which it follows that $X_t > Y_t$ eventually and hence $\P_{\CA,r}^{(x_0,y_0)}[\tau_1<\infty]=1$. 

Now assume $x_0 > y_0$. Recall that $A_{1,t}(x_0,y_0)$ is the set of paths starting from $(x_0,y_0)$ that end with the first tie at time $t$. Note that $A_{1,t}(x_0,y_0)$ is nonempty only if $t = d_0 + 2k$, where $d_0 = x_0-y_0$ and $k\geq 0$. Let $\pi\in A_{1,t}(x_0,y_0)$ and its state at time $j$ be $\pi_j = (x_j, y_j)$. The probability of the path $\pi$ is given by
\begin{align*}
\P_{\CA,r}^{(x_0,y_0)}[\pi] &= \prod_{j=0}^{t-1} \left(\frac{rx_j}{rx_j + y_j}\right)^{x_{j+1}-x_j}\left(\frac{y_j}{rx_j +  y_j}\right)^{ y_{j+1}- y_j} \\
&= \frac{r^{x_t - x_0} (x_0)_{x_t - x_0} (y_0)_{y_t - y_0}}{\prod_{j=0}^{t-1} (rx_j +y_j)} \\
&= \frac{r^{x_t - x_0} (x_0)_{x_t - x_0} (y_0)_{y_t - y_0}}{\prod_{j=0}^{t-1} [(r-1)x_j +x_0 + y_0 + j]},
\end{align*}
where in the last step we have used $x_j + y_j = x_0 + y_0 + j$. Note that $\P_{\CA,r}^{(x_0,y_0)}[\pi]$ is maximized if the $x_j$'s are minimized, subject to the constraints that the $x_j$'s increase monotonically from $x_0$ to $x_t$ with step size 0 or 1, and that $x_j > y_j$ for all $1\leq j\leq t-1$, or equivalently $x_j > x_0 + (j-d_0)/2$. This is achieved by the following sequence,
\[
x_j^* = \begin{cases}
 x_0, & j = 0,1,\dots, d_0-1;\\
 x_0 + \left\lfloor(j-d_0)/2\right\rfloor+1, & j = d_0, d_0 + 1, \dots, t - 1;\\
x_t, & j = t.
\end{cases}
\]
The corresponding path $\pi^*$ has probability
\begin{align*}
&\P_{\CA,r}^{(x_0,y_0)}[\pi^*] \\
=\;& \prod_{j=0}^{d_0-2} \frac{y_0+j}{rx_0 + y_0 + j} \left(\prod_{x = x_0}^{x_t-1} \frac{r x}{rx + (x-1)}\cdot\frac{x-1}{r(x+1)+(x-1)}\right)\\
&\times \frac{x_t - 1}{r x_t + (x_t-1)},
\end{align*}
which, after arrangement, yields,
\begin{align*}
&\P_{\CA,r}^{(x_0,y_0)}[\pi^*] \\
=\;& \prod_{j=0}^{d_0-1} \frac{y_0+j}{rx_0 + y_0 + j} \left(\prod_{x = x_0}^{x_t-1} \frac{r x}{(r+1)x+r}\cdot\frac{x}{(r+1)x + (r-1)}\right)\\
\leq\;& \frac{(y_0)_{d_0}}{(rx_0 + y_0)_{d_0}} \frac{r^{x_t-x_0}}{(r+1)^{2(x_t-x_0)}}\\
=\;& \frac{(y_0)_{d_0} (r+1)^{d_0}}{(rx_0 + y_0)_{d_0}} \left(\frac{r}{r+1}\right)^{x_t-x_0}\left(\frac{1}{r+1}\right)^{y_t-y_0}\\
=\;&\frac{(y_0)_{d_0} (r+1)^{d_0}}{(rx_0 + y_0)_{d_0}}  \P_{\RW,r}^{(x_0,y_0)}[\pi].
\end{align*}
Thus we have
\[
\P_{\CA,r}^{(x_0,y_0)}[\pi] \leq \P_{\CA,r}^{(x_0,y_0)}[\pi^*] \leq \frac{(y_0)_{d_0} (r+1)^{d_0}}{(rx_0 + y_0)_{d_0}}  \P_{\RW,r}^{(x_0,y_0)}[\pi],
\]
and, after summing over $\pi\in A_{1,t}(x_0,y_0)$,
\[
\P_{\CA,r}^{(x_0,y_0)}[\tau_1=t] \leq \frac{(y_0)_{d_0} (r+1)^{d_0}}{(rx_0 + y_0)_{d_0}}  \P_{\RW,r}^{(x_0,y_0)}[\tau_1=t].
\]
Summing over $t$, we obtain
\[
\P_{\CA,r}^{(x_0,y_0)}[\tau_1<\infty] \leq \frac{(y_0)_{d_0} (r+1)^{d_0}}{(rx_0 + y_0)_{d_0}}  \P_{\RW,r}^{(x_0,y_0)}[\tau_1<\infty].
\]
By Eq.~(3.9) and XI.3.d of \cite{fellerOne}, $\P_{\RW,r}^{(x_0,y_0)}[\tau_1<\infty] = r^{-d_0}$, from which the desired conclusion follows. 
\end{proof}

\begin{corollary}\label{cor:eventual_tie}
The probability of having at least one more tie starting from a tie state $(x,x)$ is bounded by
\[
\P_{\CA,r}^{(x,x)}[\tau_2<\infty] \leq \frac{2}{r+1}.
\]
\end{corollary}

\begin{proof}
By considering the one-step transition from $(x,x)$ into $(x,x+1)$ or $(x+1,x)$, we obtain
\begin{align*}
&\P_{\CA,r}^{(x,x)}[\tau_2<\infty] \\
=\;& \frac{r}{r+1} \P_{\CA,r}^{(x+1,x)}[\tau_1<\infty] + \frac{1}{r+1} \P_{\CA,r}^{(x,x+1)}[\tau_1<\infty]\\
\leq\;& \frac{x}{(r+1)x+r} + \frac{1}{r+1} \leq \frac{2}{r+1},
\end{align*}
where the first inequality follows from Lemma \ref{lem:eventual_tie}.
\end{proof}

Now we prove \prettyref{thm:intensity-unequal}.

\begin{proof}[of \prettyref{thm:intensity-unequal}]
Let $Z_n$ be the common value of $X_t$ and $Y_t$ at $t=\tau_n$, i.e., $Z_n = X_{\tau_n}$. Conditioned on $\tau_n<\infty$ and $Z_n=z$, the probability of $\tau_{n+1}<\infty$ is just the probability of having a tie after leaving $(z, z)$. Thus
\begin{align*}
 \P_{\CA,r}^{(x_0,y_0)}[\tau_{n+1}<\infty \mid \tau_n < \infty, Z_n=z] 
=\P_{\CA,r}^{(z,z)}[\tau_2<\infty]\leq \frac{2}{r+1},
\end{align*}
by \prettyref{cor:eventual_tie}.
Removal of the conditioning yields
\[
\P_{\CA,r}^{(x_0,y_0)}[\tau_{n+1}<\infty \mid \tau_n < \infty] \leq \frac{2}{r+1}.
\]
It follows that
\begin{align*}
& \P_{\CA,r}^{(x_0,y_0)}[N\geq n] = \P_{\CA,r}^{(x_0,y_0)}[\tau_n <\infty] \\
=\;& \P_{\CA,r}^{(x_0,y_0)}[\tau_1 <\infty] \prod_{i=1}^{n-1}\P_{\CA,r}^{(x_0,y_0)}[\tau_{i+1}<\infty\mid \tau_{i}<\infty]\\
\leq\; & \P_{\CA,r}^{(x_0,y_0)}[\tau_1<\infty] \left(\frac{2}{r+1}\right)^{n-1}.
\end{align*}
An application of Lemma \ref{lem:eventual_tie} completes the proof.
\end{proof}

%% file: TEX/conclusion.tex

\section{Discussion and Conclusion} \label{sec:conclusion}

As supported by various empirical studies over the last century, real world competitions for resource 
accumulation seem to be subject to cumulative advantage effects, at least to some degree. Because of such findings 
and the fact that CA generally leads to large inequalities in resource distribution, understanding the 
role of skill and luck in competition dynamics becomes a pressing issue both in theory and practice. 
Indeed, recent empirical~\cite{Salganik2006,van2014field} and 
theoretical~\cite{denrell2012top,Antal10,krapivsky2002statistics} studies have contributed in this direction. 
However, contrary to prior theoretical works, we consider simple and classical mathematical models 
that capture just the essence of skill and luck competitions with and without CA effects, and investigate 
fundamental aspects of competition, namely duration (i.e., time until ultimate winner emerges) and intensity 
(i.e., number of ties in competition). By considering simple models and simple properties we prove 
and illustrate fundamental theoretical results: CA effect exacerbates the role of luck - power law 
tail duration emerges regardless of skill differences, and become extreme (i.e., infinite mean) when 
skill differences are small enough. Moreover, duration is long not necessarily because of intense competition 
where agents tussle aggressively for ultimate leadership. On the contrary, under CA, competitions 
are generally very mild, exhibiting an exponential tail. Long competitions 
emerge when an early stroke of luck places the less skilled in the lead, who can then, boosted by CA effects, 
 enjoy leadership for a very long period of time. Thus, when CA is present luck sides 
with the less skilled. 

The non-negligible probability of long-lasting competitions has far-reaching implications. 
In the absence of CA, it takes very little time for the fittest agent to establish dominance, 
so it is often reasonable to neglect the possibility of a premature burnout. Such observations 
are in hand with the ``survival of the fittest'' principle, since soon enough the more skilled 
will prevail. In the presence of CA, however, even agents with large fitness superiority may face 
the challenge of having to endure extremely long competitions. This challenge becomes all more 
real when the fitness superiority is only minimal. Will the more skilled 
survive the seemingly eternal inferiority during the competition? Under CA time becomes a 
central issue, with delusion becoming reality if the more skilled burns out during a long struggle. 
Thus, in the face of CA, the fittest survives only if it can persist, which prompts us to rename the principle 
``survival of the fittest and persistent'' when considering CA competitions.